\documentclass[a4paper,10pt]{article}
\usepackage[margin=2cm]{geometry}
\usepackage{amsmath, amsthm, amsfonts}
\usepackage{xspace}
\usepackage[dvipsnames]{xcolor}
\usepackage{url}
\usepackage{forest}
\usepackage{wrapfig}
\usepackage{subfig}
\usepackage{thmtools}
\usepackage{enumitem}
\usepackage{tikz}
\usetikzlibrary{matrix,positioning,shapes.callouts,tikzmark, decorations.pathreplacing,arrows,calc,chains,patterns,patterns.meta}

\usepackage{bbding}

\usepackage{epstopdf}
\DeclareGraphicsExtensions{.eps}

\usepackage{floatrow}
\newfloatcommand{capbtabbox}{table}[][\FBwidth]

\renewcommand{\epsilon}{\varepsilon}

\newtheorem{theorem}{Theorem}
\newtheorem{corollary}[theorem]{Corollary}
\newtheorem{lemma}[theorem]{Lemma}

\newtheorem{definition}{Definition}

\definecolor{niceblue}{rgb}{.392,.584,.929}
\definecolor{nicered}{rgb}{1,.375,.375}
\definecolor{nicegreen}{rgb}{.1,.8,.4}
\definecolor{nicepurple}{rgb}{.78,.26,.93}
\definecolor{niceorange}{rgb}{1,.575,.25}
\definecolor{niceyellow}{rgb}{.98,.94,.3}

\newcommand{\defn}{\emph}
\newcommand{\polylog}{\text{polylog}}
\newcommand{\poly}{\text{poly}}
\renewcommand{\hat}{\widehat}
\renewcommand{\tilde}{\widetilde}
\DeclareMathOperator{\E}{E}

\newif\iffull
\fulltrue

\newif\ifanon
\anonfalse

\begin{document}

\title{Improved Space-Efficient Approximate Nearest Neighbor Search Using Function Inversion}

\ifanon
\author{Anonymous Submission}
\else
\author{Samuel McCauley\\
Williams College\\
Williamstown, MA 01267 USA\\
\url{sam@cs.williams.edu}}
\fi

\date{}

\pagenumbering{gobble}
\maketitle
\begin{abstract}
	Approximate nearest neighbor search (ANN) data structures have widespread applications in machine learning, computational biology, and text processing.  The goal of ANN is to preprocess a set $S$ so that, given a query $q$, we can find a point $y$ whose distance from $q$ approximates the smallest distance from $q$ to any point in $S$.  
For most distance functions, the best-known ANN bounds for high-dimensional point sets are obtained using techniques based on locality-sensitive hashing (LSH).  

Unfortunately, space efficiency is a major challenge for LSH-based data structures.  Classic LSH techniques require a very large amount of space, oftentimes polynomial in $|S|$.  A long line of work has developed intricate techniques to reduce this space usage, but these techniques suffer from downsides: they must be hand tailored to each specific LSH, are often complicated, and their space reduction comes at the cost of significantly increased query times.

In this paper we explore a new way to improve the space efficiency of LSH using function inversion techniques, originally developed in (Fiat and Naor 2000).  

We begin by describing how function inversion can be used to improve LSH data structures.  This gives a fairly simple, black box method to reduce LSH space usage.  

Then, we give a data structure that leverages function inversion to improve  the query time of the best known near-linear space data structure for approximate nearest neighbor search under Euclidean distance: the ALRW data structure of (Andoni, Laarhoven, Razenshteyn, and Waingarten 2017).  
ALRW was previously shown to be optimal among ``list-of-points'' data structures for both Euclidean and Manhattan ANN;\@
thus, in addition to giving improved bounds, our results imply that list-of-points data structures are not optimal for Euclidean or Manhattan ANN\@.
\end{abstract}

\newpage
\pagenumbering{arabic}  
\setcounter{page}{1}
\section{Introduction}%
\label{sec:introduction}

Modern data science relies increasingly on sophisticated ways to query datasets.  A fundamental class of these queries is similarity search: given a query $q$, find the item in the dataset that is most ``similar'' to $q$.  Similarity search was originally introduced by Minsky and Papert in their seminal textbook~\cite{MinskyPapert69}.  Similarity search problems have applications in wide-ranging areas, including clustering~\cite{Berkhin06,HachenbergGottron13}, pattern recognition~\cite{CoverHart67}, data management~\cite{SaltonWoYa75}, and compression~\cite{KhanDoAn17}.

Unfortunately, for many types of similarity search queries, all known approaches require space or query time exponential in the dimension of the data~\cite{HarPeledInMo12}.  This is often called the curse of dimensionality.  For datasets with high---say $\Omega(\log n \log\log n)$---dimensions, this cost quickly becomes infeasible.  
Fine-grained complexity results reinforce this barrier: sublinear-time algorithms for many similarity search problems would imply that the Strong Exponential Time Hypothesis is false~\cite{AlmanWilliams15,Rubinstein18,Williams18}.

Therefore, to obtain good theoretical guarantees, a long line of research has focused on approximate similarity search~\cite{IndykMotwani98,HarPeledInMo12}.  Rather than giving the most similar point, an approximate similarity search data structure simply guarantees a point whose similarity approximates the similarity of the most similar point.

Specifically, in this work we focus on \defn{Approximate Nearest Neighbor search (ANN)}.  In ANN, the data structure preprocesses a set $S$ (a subset of a universe $U$) of $n$ $d$-dimensional points for a distance function $d(\cdot, \cdot)$, a radius $r$, and an approximation factor $c> 1$.  On a query $q\in U$, the data structure gives the following guarantee: if there exists a point $x\in S$ with $d(q, x)\leq r$, then with probability at least $.9$, the data structure returns a point $y\in S$ with $d(q, y)\leq cr$.  The goal is to obtain solutions parameterized by the approximation ratio $c$, obtaining polynomial performance in $n$ and $d$ for any constant $c>1$.

Classic results show that the definition of ANN given above can be generalized without significantly increasing the cost.
The correctness guarantee of $.9$ is arbitrary; creating independent copies of the structure can drive the success probability arbitrarily close to $1$.  Past results allow $S$ to change dynamically or to work without knowing the value of $r$ up front~\cite{IndykMotwani98}, and to obtain all near neighbors (rather than just one)~\cite{AhleAuPa17,Indyk01}.
Moreover, the approximation guarantee of ANN can be viewed as a beyond-worst-case guarantee: if the dataset is well-spaced so that there is a single point within distance $cr$ of the query, ANN guarantees that it will be returned (this was generalized further in~\cite{Datar04}).  This may explain in part why ANN solutions are effective at finding exact nearest neighbors in practice, as was seen in~\cite{Datar04,AndoniInLa15}.

\paragraph{ANN Distance Functions.}
The most widely investigated ANN problems are Euclidean and Manhattan ANN.  Euclidean and Manhattan ANN have $U = \mathbb{R}^d$; Euclidean ANN uses the standard $\ell_2$ distance function for $d(\cdot, \cdot)$, while Manhattan ANN uses $\ell_1$ for $d(\cdot, \cdot)$.

There are several motivations for studying Euclidean ANN in particular.
Euclidean distance 
is a natural and well-known metric. 
It is particularly useful for real-world similarity search problems~\cite{AndoniInLa15}.
Furthermore, Euclidean ANN can be used to solve other problems.
For example, Manhattan ANN can be solved using a Euclidean ANN data structure using a classic embedding~\cite{AndoniLaRa17,LinialLoRa95}.  In fact, this embedding gives the state-of-the-art Manhattan ANN bounds~\cite{AndoniLaRa17}.  Euclidean ANN is similarly used in the state-of-the-art method for finding the closest point under cosine similarity~\cite{AndoniInLa15}.

\paragraph{Locality-Sensitive Hashing for ANN.}
Many theoretical results for ANN are based on \defn{locality-sensitive hashing}, originally developed in~\cite{IndykMotwani98,HarPeledInMo12}.  An LSH is a hash function where points at distance $r$ hash to the same value with probability at least $p_1$, while points at distance more than $cr$ hash to the same value with probability at most $p_2$.  Such a hash immediately implies an ANN data structure with query time $\Theta(n^{\rho})$ and space $\Theta(n^{1 + \rho})$, where $\rho$ is defined as $\rho = \log p_1 / \log p_2$---see Section~\ref{sec:preliminaries} for a full exposition.

Locality-sensitive hashes have been developed for many distance functions, including Jaccard similarity~\cite{Broder97,ChristianiPagh17}, Edit Distance~\cite{McCauley21,MarcaisDePa19}, Frechet distance~\cite{DriemelSilvestri17}, and $\chi^2$ distance~\cite{GorisseCoPr11}, among many others.  

\paragraph{Space-Efficient ANN Data Structures.}
Despite its popularity and applicability, one downside of LSH stands out: an LSH-based data structure requires $\Omega(n^{1 + \rho})$ space.  This space usage quickly becomes prohibitive.  For example, for $\ell_1$, the classic LSH of Sar-Peled, Indyk, and Motwani~\cite{IndykMotwani98,HarPeledInMo12} obtains $\rho = 1/c$, so if $c=2$ the data structure requires $\Omega(n^{3/2})$ space.  

A long line of research has investigated ways to improve the space usage for LSH; much of this work focused on the Euclidean distance~\cite{Panigrahy06,Andoni09,Kapralov15,Christiani17,AndoniLaRa17,SantoyoChTe13,AumullerPaSi20,AumullerHaMa21}.  Of particular interest is the near-linear-space regime, where the space required by the data structure is close to $O(nd)$.

Space-efficient LSH approaches are difficult to design.
This is because, at a high level, space-efficient ANN methods usually store points based on a space partition much like in an LSH; the data structure saves space by storing the points in fewer locations.
The data structure then must probe more locations on each query to ensure correctness.  
Designing a correct approach along these lines---ensuring that the query probes in the correct locations to guarantee correctness---is technically challenging, and each approach must be carefully tailored to the specific space partition being used (see the discussion of techniques and related work in~\cite{Kapralov15} for example).  
Perhaps for this reason, many ANN problems have no known space-efficient solutions with theoretical query guarantees.

\paragraph{Best Known Bounds.}
Andoni et al.~\cite{AndoniLaRa17} obtain the current state-of-the-art bounds for space-efficient Euclidean ANN and Manhattan ANN\@.  Their results include a smooth tradeoff between the space usage and query time of the data structure, achieving state-of-the-art bounds along the entire curve.

Interestingly, there is a matching conditional lower bound given in~\cite{AndoniLaRa17}.  They define a type of data structure, a \defn{list-of-points data structure}, to encompass ``LSH-like'' approaches.  In short, a list of points data structure explicitly stores lists of points; each query must choose a subset of lists to look through for a point at distance $\leq cr$.  Most high-dimensional ANN data structures are list-of-points data structures, and Andoni et al. conjecture that their lower bound can be generalized to handle the few exceptions~\cite{AndoniLaRa16}.  

The data structures presented in this paper are not list-of-points data structures, as they do not store lists of points explicitly.  Instead, our data structures implicitly store lists of points while retaining good query time using a sublinear-space function inversion data structure (defined below).  This will allow us to improve their query time, breaking the list of points lower bound.

\paragraph{Function Inversion.}
Our results work by applying function inversion to LSH.  Function inversion data structures were initially proposed by Hellman~\cite{Hellman80}, and later analyzed by Fiat and Naor~\cite{FiatNaor00}.  In short, these data structures allow us to preprocess a function $f: \{1, \ldots, N\}\rightarrow \{1,\ldots, N\}$ using $o(N)$ space so that for a given $y\in \{1, \ldots, N\}$, we can find an $x$ with $f(x) = y$ in $o(N)$ time (see Lemma~\ref{lem:fiatnaorworst} for specifics).

A recent, exciting line of work has looked at how these function inversion data structures can be applied to achieve space- and time-efficient solutions to classic data structure problems.  Function inversion can be applied to an online 3SUM variant to give new time-space tradeoffs~\cite{KopelowitzPorat19,GolovnevGuHo20}.  
Aronov et al.~\cite{AronovEzSh23} gave results for the closely-related problem of collinearity testing.
Bille et al.~\cite{BilleGoLe22} use the 3SUM method as a black box to give improved string indexing methods.
Finally, Aronov et al.~\cite{AronovCaDa23} recently gave a toolbox for using function inversion for implicit set representations, with a number of applications.

We show that ANN data structures interface particularly well with function inversion.  Specifically, in Section~\ref{sec:function_inversion_for_locality_sensitive_hashing}, we take advantage of the repeated functions of LSH: since we want to invert many functions, we can store extra metadata (to be shared by all functions) to help queries.  Meanwhile, in Section~\ref{sec:faster_euclidean_linear_space_ann}, we will show how to combine the data structure of Andoni et al.~\cite{AndoniLaRa17}, with function inversion to achieve improved Euclidean and Manhattan ANN performance.

\subsection{Results}%
\label{sec:results}

Two of the most significant questions remaining in the area of space-efficient ANN are:
\begin{enumerate}[topsep=2pt,noitemsep]
	\item Is it possible to obtain a black-box method to improve the space usage of an LSH-based ANN data structure?
	\item Is it possible to improve the query time for space-efficient Euclidean ANN beyond the lower bound for list-of-points data structures?
\end{enumerate}
Our results give positive answers to these two questions.

First, we show how to use function inversion to improve the space efficiency of any locality-sensitive hashing method.  

\begin{restatable}{theorem}{basiclsh}
\label{thm:basic_lsh}
For any locality-sensitive hash family $\mathcal{L}$ 
(where $\mathcal{L}$ has $\rho = \log p_1/\log p_2$ and evaluation time $T$, and storing a given $\ell\in \mathcal{L}$ requires $O(n^{1-\rho})$ space), approximation ratio $c$, and  space-saving parameter $s < \rho$, 
there exists an ANN data structure with preprocessing time $O(n^{1 + \rho})$, expected space $\tilde{O}(n^{1 + \rho - s})$ and expected query time $\tilde{O}(T n^{\rho + 3s})$.
\end{restatable}

This black box method is the first space-efficient method for many ANN problems.  Even for the well-studied Euclidean ANN problem, our simple black-box method can be combined with the classic LSH of Andoni and Indyk~\cite{AndoniIndyk06} to give bounds that are competitive with, or even improve upon, previous algorithms (see Table~\ref{tab:eucANNhistory}).  

Second, we show how to use function inversion to obtain a data structure that gives improved state-of-the-art query times for near-linear-space Euclidean ANN and Manhattan ANN.

\begin{restatable}{theorem}{runningtime}%
\label{thm:running_time}
There exists a Euclidean ANN data structure requiring $n^{1 + o(1)}$ space and at most $n^{1.013 + o(1)}$ preprocessing time that can answer queries in $n^{\alpha(c) + o(1)}$ expected time with 
\[
	\alpha(c) = \frac{2c^2 - 1}{c^4}\left(1 - \frac{(c^2 - 1)^2}{4c^4 + (c^2 - 1)^2}\right).
\]
\end{restatable}

In short, the idea of our data structure is to take the data structure of Andoni et al.~\cite{AndoniLaRa17} (which we refer to as ALRW), and replace the lists of points with a function-inversion data structure.  Unfortunately, ALRW itself---even without the lists---requires far too much space.  Thus, achieving our bounds requires  trimming the ALRW data structure and carefully running function inversion on the result.

\subsection{Comparing Results}%
\label{sec:comparing_results}

\begin{wraptable}{r}{.4\textwidth}
	\renewcommand*{\arraystretch}{1.5}
	\begin{tabular}{|c|c|}
	\hline 
	Data Structure  & Query Time Exponent \\
	\hline 
	\hline 
	\cite{Panigrahy06} & $2.06/c$ \\
	\hline 
	\cite{Andoni09,AndoniIndyk06} & $O(1/c^2)$ \\
	\hline 
	Thm~\ref{thm:basic_lsh} w/~\cite{AndoniIndyk06} & $4/c^2$ \\
	\hline 
	\cite{Kapralov15} & $4/(c^2 + 1)$ \\
	\hline 
	\cite{Christiani17} & $4c^2/(c^2 + 1)^2$ \\
	\hline 
	ALRW~\cite{AndoniLaRa17} & $(2c^2 - 1)/c^4$ \\
	\hline 
	Thm~\ref{thm:running_time} & $\frac{2c^2 - 1}{c^4}\left(1 - \frac{(c^2 - 1)^2}{4c^4 + (c^2 - 1)^2}\right)$ \\
	\hline 
\end{tabular}
\caption{This table gives a summary of known results for Euclidean ANN data structures using $n^{1 + o(1)}$ space. A data structure with listed query time exponent $q$ indicates a query time of at most $n^{q}$; lower-order terms are not listed.}%
\label{tab:eucANNhistory}
\end{wraptable}

In this section we compare our results to past work, focusing specifically on Euclidean ANN.  In Table~\ref{tab:eucANNhistory} we give a brief summary of the query time achieved by past near-linear (i.e.\ $n^{1 + o(1)}$ space) data structures.

We focus specifically on giving intuition on how our query time improves on the state of the art.
The previous state of the art query time by Andoni et al.~\cite{AndoniLaRa17}; we refer to it as $n^{ALRW(c) + o(1)}$, with $ALRW(c) = (2c^2 - 1)/c^4$.  In Theorem~\ref{thm:running_time} we achieve query time $n^{\alpha(c) + o(1)}$ with $\alpha(c) = \frac{2c^2 - 1}{c^4}\left(1 - \frac{(c^2 - 1)^2}{4c^4 + (c^2 - 1)^2}\right)$.

Since $(c-1)^2$ and $4c^2$ are positive, we immediately have $\alpha(c) < ALRW(c)$.  Furthermore, as $c$ gets large, $(c^2 - 1)^2/(4c^4 + (c^2 - 1)^2)$ approaches  $1/5$; more concretely, 
one can verify that if $c > 2.6$ we have $\alpha(c) < .85ALRW(c)$.

\begin{figure}[t]
	\begin{floatrow}
		\ffigbox{%
	\begin{tikzpicture}
		\node[inner sep=0pt] (plot) at (0,0) {\includegraphics[width=.8\columnwidth]{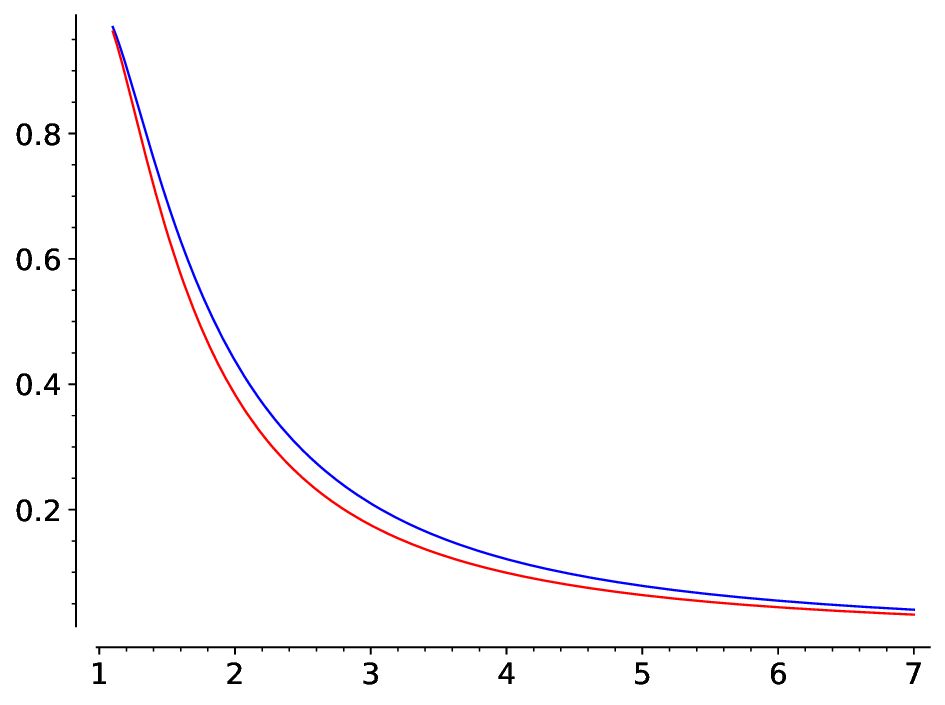}};
		\node[above=-.2cm of plot] {\textbf{Comparing Theorem~\ref{thm:running_time} and $ALRW(c)$}};
		\node[below right=-.01cm and -3cm of plot,anchor=north] {Approximation factor $c$};
		\node[left=-.01cm of plot,rotate=90,anchor=south] {Query time exponent};
		\node[above right=-1.5cm and -2cm of plot,draw,rectangle,align=left] {\footnotesize \tikz[baseline=-2pt]{\draw[color=red] (0,0) to (.3,0);} $\alpha(c)$ \\ \footnotesize \tikz[baseline=-2pt]{\draw[color=blue] (0,0) to (.3,0);} $ALRW(c)$};
	\end{tikzpicture}
		}{%
	\caption{A figure comparing $\alpha(c)$ from Theorem~\ref{thm:running_time} to $ALRW(c)$.  The $y$-axis represents the exponent of the query time: our results obtain a linear-space data structure with query time $n^{\alpha(c) + o(1)}$, compared to the state of the art in~\cite{AndoniLaRa17} with query time $n^{ALRW(c) + o(1)}$.}%
\label{fig:comparison}
}%
\hspace{.3in}
\capbtabbox{%
\begin{tabular}{ c|c c  c}
	$c$  & $\alpha(c)$  & $ALRW(c)$ & preproc.\\
	\hline 
	1.05 & .989 & .991 & {1.001}\\
	1.5 & .641 & .691 &  {1.011}\\
	1.79 & .471 & .527 & {1.012}\\
	2 & .383 & .438 &    {1.012}\\
	3 & .175 & .210 &    {1.007} \\
	10 & .016 & {.020} & {1.001}\\
\end{tabular}
\vspace{1.5cm}
}{%
	\caption{A table comparing the exponent of the query time of our linear-space approach vs that of~\cite{AndoniLaRa17}.  All values are rounded to the third decimal place.  In the final column, we give the exponent of the preprocessing time\iffull (e.g. $n^{1.011 + o(1)}$ time for $c=1.5$)\fi.}%
\label{tab:comparison}%
}%
\end{floatrow}
\vspace{-.1in}
\end{figure}

For a more complete picture, we compare the bounds obtained by each in Figure~\ref{fig:comparison} and Table~\ref{tab:comparison}.  Specifically, this gives a sense of how much improvement is obtained for various values of $c$.  
We see that for values of $c$ close to $2$ we obtain a query time improvement of roughly $n^{.04}$. 
In Table~\ref{tab:comparison}, we also give the exact exponent of the preprocessing time; for large or small $c$ this is noticeably better than the $n^{1.013}$ upper bound from Theorem~\ref{thm:running_time}.  
While our improved query time comes at a cost of increased preprocessing time (ALRW requires $n^{1 + o(1)}$ preprocessing time), this increase is fairly mild.
We include an entry at $c = 1.79$ as it maximizes $ALRW(c) - \alpha(c)$.

\subsection{Related Work}%
\label{sec:related_work}

\paragraph{ANN.}
We briefly describe past work on ANN.  See the survey by Andoni and Indyk~\cite{AndoniIndyk17} for a more thorough exposition.

As mentioned above, theoretical bounds for ANN without exponential dependance on $d$ usually rely on locality-sensitive hashing.  LSH-based methods (including data-dependent results) have seen extensive work on Euclidean ANN in particular~\cite{HarPeledInMo12,Datar04,AndoniIndyk06,AndoniInNg14,AndoniRazenshteyn15,AndoniLaRa17,Ahle17}.  Further work has applied LSH to many other metrics~\cite{ChristianiPagh17,McCauley21,MarcaisDePa19,DriemelSilvestri17,GorisseCoPr11,AhlePaRa16}; this includes the well-known (and widely used) Bit Sampling LSH (for Hamming distance)~\cite{HarPeledInMo12} and MinHash (for Jaccard similarity)~\cite{Broder97}.

In terms of practical performance, LSH-based methods are competitive, though they can be outperformed by heuristic methods on structured data (for a full comparison see e.g. the benchmark of Aum\"{u}ller et al.~\cite{AumullerBeFa20}).  
There has been work on giving worst-case bounds for the performance of these heuristics on certain types of datasets~\cite{MccauleyMiPa18,Laarhoven18,AndoniInLa15,AhleAuPa17}.  

\paragraph{Space-Efficient ANN.}
Theoretical bounds in the low-space regime have focused largely on Euclidean ANN.

Panigrahy gave a data structure for Euclidean ANN using $\tilde{O}(n)$ space, and achieving query time at most $\tilde{O}(n^{2.09/c})$ (this bound can be improved for small $c$; e.g.\ if $c=2$ the query time is roughly $n^{.69}$)~\cite{Panigrahy06}.

Andoni and Indyk gave $\tilde{O}(n)$-space data structure with an improved query time of $n^{O(1/c^2)}$, where the constant in the exponent of $n$ is not specified~\cite{AndoniIndyk06}; this is discussed further in Andoni's thesis~\cite{Andoni09}.

Kapralov gave the first data structure that can trade off smoothly between space and time~\cite{Kapralov15}.  Setting the parameters to $O(dn)$ space\iffull{} (linear in the size of the input data)\fi, this data structure requires $n^{4/(c^2 + 1) + o(1)}$ query time.  

A sequence of followup papers~\cite{Laarhoven15,Christiani17,AndoniLaRa16b} (see also~\cite{AumullerPaSi20}) 
culminated in Andoni et al.~\cite{AndoniLaRa17} giving a data structure achieving smooth tradeoffs between time and space for Euclidean and Manhattan distances, along with a matching lower bound for ``list-of-points'' data structures.  Setting the space to be $n^{1 + o(1)}$, they achieve query time $n^{(2c^2 - 1)/c^4 + o(1)}$.  

\paragraph{Function Inversion.}

The problem of giving smooth time-space tradeoffs for inverting a black-box function was initiated by Hellman~\cite{Hellman80}.  This idea was generalized and improved by Fiat and Naor~\cite{FiatNaor00}, and recently improved further by Golovnev et al.~\cite{GolovnevGuPe23}.
There are known lower bounds as well: Hellman's original result is nearly optimal for inverting random functions when limited to a restricted class of algorithms~\cite{BarkanBiSh06}, and there are tight upper and lower bounds for inverting permutations~\cite{Yao90}.

As mentioned earlier, function inversion has recently been applied to data structures, achieving improved tradeoffs for 3SUM~\cite{KopelowitzPorat19,GolovnevGuHo20}, as well as the related problems of collinearity testing~\cite{AronovEzSh23} and string indexing~\cite{BilleGoLe22}.  The results in this paper extend this line of work, applying function inversion to ANN\@.

\section{Preliminaries}%
\label{sec:preliminaries}

\paragraph{Locality-Sensitive Hashing.}

We begin by formally defining a locality-sensitive hash family.
\begin{definition}
	A hash family $\mathcal{H}$ is $(p_1, p_2, r, cr)$-sensitive for a distance function $d(\cdot, \cdot)$ if for any points $x$, $y$ with $d(x,y)\leq r$, $\Pr_{h\in\mathcal{H}}[h(x) = h(y)] \geq p_1$, and for any points $x'$, $y'$ with $d(x', y') \geq cr$, $\Pr_{h\in\mathcal{L}}[h(x') = h(y')] \leq p_2$.
\end{definition}

We describe how a locality-sensitive hash can be used for ANN search as was originally described by Indyk and Motwani~\cite{IndykMotwani98,HarPeledInMo12}.  

To begin, we select a sequence of hash functions.
The first step is to create a new LSH by concatenating $\lceil \log_{1/p_2} n\rceil$ independently-chosen functions from $\mathcal{H}$; let $\mathcal{L}$ be the family of possible hash functions resulting from this concatenation.  Thus, if $d(x',y') \geq cr$, $\Pr_{\ell\in \mathcal{L}}[\ell(x') = \ell(y')] \leq 1/n$.   These parameters are set so that there is one expected point $z$ with $d(z, q) \geq cr$ and $\ell(z) = \ell(q)$.  

ANN performance is generally given in terms of $\rho = \log p_1/\log p_2$.  
We have for $x$, $y$ with $d(x,y)\leq r$, $\Pr_{\ell}[\ell(x) = \ell(y)] \geq p_1/n^{\rho}$.
Thus, we repeat the above steps $R = \Theta(n^{\rho}/p_1)$ times to obtain $R$ hash functions $\ell_1, \ldots, \ell_R$, each sampled independently from $\mathcal{L}$.  

Now, preprocessing.
To preprocess, 
create a \defn{reverse lookup table} for $\ell_i$ for all $i \in \{1,\ldots,R\}$.
A reverse lookup table is a key-value store (for example, we can implement it using a hash table).  For each $y$ such that there exists an $x$ with $\ell_i(x) = y$, the reverse lookup table stores $y$ as the key, and $\{x ~|~ \ell(x) = y\}$ as the value.  Thus, each reverse lookup table takes $\Theta(n)$ space (giving $\Theta(n^{1+\rho}/p_1)$ space in total), and using these tables we can find $\ell^{-1}_j(q)$ for a given $q$ and $i$ in $O(1 + |\ell^{-1}_i(q)|)$ expected time.

To perform a query $q$, we look up $\ell^{-1}_i(q)$ for all $i\in \{1,\ldots, R\}$, using the reverse lookup table for each hash function.  We call all points found this way \defn{candidate points}; there are $O(n^{\rho}/p_1)$ candidate points with distance more than $cr$ from $q$ in expectation.  (If we find a point at distance less than $cr$ we return it, so such points only increase the number of candidate points by an additive $O(1)$.)  

If $x$ has $d(x,q) \leq r$, then the query fails only if $x$ is not in the candidate points.  This occurs with probability $(1 - p_1/n^{\rho})^{R}$; this is at most $.1$ if the constant when setting $R = \Theta(n^{\rho}/p_1)$ is sufficiently large.

The space usage of the data structure is $O(nR) = O(n^{1+\rho}/p_1)$ for the reverse lookup tables, plus the space to store $\ell_1, \ldots, \ell_R$.  (Usually a particular function $\ell_i$ can be stored in $n^{o(1)}$ space, so this is a lower-order term.) A query takes $O(1)$ expected time for each reverse hash table lookup if $\ell_i$ can be evaluated in constant time, for $O(R) = O(n^{\rho}/p_1)$ time overall; the query time increases linearly if an $\ell_i$ sampled from $\mathcal{L}$ is slower to evaluate in expectation.

\paragraph{Data-dependent ANN Solutions.}

Locality-sensitive hashing as defined above is \defn{data-independent}: we choose the hashes independent of $S$.  

Interestingly, it is possible to obtain improved ANN bounds using \defn{data-dependent} techniques~\cite{AndoniInNg14}: the point set is partitioned into lists much like LSH, but these partitions are chosen using, in part, properties of the data itself.  
Thus, generating the data structure requires making random choices that depend not only on $n$, $r$, and $c$, but also properties of the point set $S$ itself.  
The state of the art ANN bounds for Euclidean space use data-dependent techniques~\cite{AndoniLaRa17}.

\paragraph{Definitions.}

Throughout the paper, we assume the dataset $S$ is in some arbitrary order; i.e.\ $S = x_1, x_2, \ldots, x_n$.  We assume without loss of generality that $n$ is a power of $2$.

We use $\tilde{O}(f(n))$ to represent $O(f(n)\polylog n)$.  Many of our bounds have an $n^{o(1)}$ term; since $\polylog n = n^{o(1)}$, we generally drop $O$ and $\tilde{O}$ notation when this term is present.  We say that a data structure is \defn{near-linear-space} if it requires space $n^{1 + o(1)}$ for a dataset of size
$n$.

We use $\circ$ to denote concatenation, and $[N]$ to represent $\{1,\ldots, N\}$.  For any list $A$, we use $A[i]$ to denote the $i$th element of $A$ ($0$-indexed).

For any function $f: X \rightarrow Y$, we call $X$ the \defn{domain} and $Y$ the \defn{codomain} of $f$.  We define $f^{-1}(y) = \{x \in X ~|~ f(x) = y\}$.  

We say that an event occurs \defn{with high probability} if, for any $C$, it occurs with probability bounded below by $1 - 1/n^C$.  (Generally, the event is parameterized by a constant $C_2$, usually hidden in $O$ notation; adjusting $C_2$ according to the desired $C$ can achieve the probability bound---see Corollary~\ref{cor:chernoff_log} for example.)

If an event occurs with high probability then we assume it occurs.  (If 
any with high probability event does not occur, we revert to a trivial data structure that scans $S$ for each query; this increases the expected query cost by $o(1)$.)

\paragraph{Formulas.}
We use the following well-known formulas.  First, $\binom{x}{y} \leq (ex/y)^y$.  Furthermore, $(1 + 1/n)^n \leq e$, and $(1 - 1/n)^n \leq 1/e$.

We use Chernoff bounds throughout the proofs; we reiterate them here for completeness.

\begin{lemma}[\cite{Chernoff52,MitzenmacherUpfal17}]%
\label{lem:chernoff}
Let $X = X_1 + X_2 + \ldots$ be the sum of identical independent $0/1$ random variables.  Then for any $\delta \geq 0$, $ \Pr[X \geq (1 + \delta)] \leq e^{-\delta^2 \E[X]/(2 + \delta)} $, and for any $1 \geq \delta \geq 0$, $\Pr[X \leq (1 - \delta)] \leq e^{-\delta^2 \E[X]/2}$.
\end{lemma}

The following parameter setting is particularly useful.

\begin{corollary}%
\label{cor:chernoff_log}
	Let $X = X_1 + X_2 + \ldots$ be the sum of identical independent $0/1$ random variables with $\E[X] = \Omega(\log n)$.  Then $X = \Theta(\E[X])$ with high probability.
\end{corollary}

We refer to the following bound as the \defn{union bound} (sometimes called Boole's inequality): for any $k$ events $E_1, \ldots, E_k$, we have $\Pr[E_1 \text{ or }\ldots \text{ or } E_k] \leq \Pr[E_1] + \ldots + \Pr[E_k]$.

\section{Function Inversion for LSH}%
\label{sec:function_inversion_for_locality_sensitive_hashing}

In this section we give our basic function inversion data structure and apply it to LSH\@.

\subsection{Function Inversion}%
\label{sec:function_inversion}

The basic building block of this paper is the function inversion data structure of Fiat and Naor~\cite{FiatNaor00}, which gave theoretical bounds for time-space tradeoffs to invert any function $f: [N]\rightarrow[N]$.

\begin{lemma}[\cite{FiatNaor00}]%
\label{lem:fiatnaorworst}
For any function $f:[N]\rightarrow [N]$ that can be evaluated in $T(f)$ time, and any $\sigma > 0$, there exists a data structure that requires $\tilde{O}(N/\sigma)$ space that can, for any $q$, find an $x$ such that $f(x) = q$ with constant probability in 
$\tilde{O}(T(f) \sigma^3)$ time.
\end{lemma}

This result was recently generalized by Golovnev, Guo, Peters, and Stephens-Davidowitz to give improved bounds for large $\sigma$, namely, query time $\tilde{O}(T(f)\min\{\sigma^3, \sigma\sqrt{N}\})$~\cite{GolovnevGuPe23}.  However, their model allows either access to a large random string that does not count toward the space bound, or a large fixed advice string for each $N$. In contrast, Lemma~\ref{lem:fiatnaorworst} in~\cite{FiatNaor00} uses explicit hash functions.  It is plausible that the results of~\cite{GolovnevGuPe23} would work in such a setting as well; this would immediately improve Theorems~\ref{thm:basic_lsh} and~\ref{thm:fiatnaor_all} for large $\sigma$.

\paragraph{Obstacles to Applying Function Inversion to LSH.} 
Function inversion has immediate potential to improve LSH performance: rather than storing a reverse lookup table to find $\ell_i^{-1}(q)$, we can use function inversion to find $\ell_i^{-1}(q)$ instead.  We obtain the same candidate points, so correctness is guaranteed, but we improve the space usage by a factor $\tilde{\Theta}(\sigma)$. (We describe this strategy in more detail in Section~\ref{sec:function_inversion_for_lsh} below.)

However, Lemma~\ref{lem:fiatnaorworst} cannot be immediately applied to LSH\@.  The first issue is that the codomain $D$ of an LSH is unlikely to be $[N]$.  This can be handled (in short) by hashing the output using a hash function $h: D\rightarrow [N]$ from a universal hash family; this technique is standard in the literature (see~\cite{AronovCaDa23,CorriganGibsKogan19,GolovnevGuHo20,KopelowitzPorat19,GolovnevGuPe23}).

Even after reducing the codomain of the function, Lemma~\ref{lem:fiatnaorworst} does not suffice for our purposes.  When using locality-sensitive hashing, we must compare the query to \emph{all} points in its preimage, whereas Lemma~\ref{lem:fiatnaorworst} only returns a single point.  In the remainder of this subsection, we show in fact we can obtain all $x$ with $f(x) = q$.  Note that this result requires an additive $O(N)$ space---it is only useful when we want to invert multiple hash functions over the same set $[N]$.

\paragraph{The All-Function Inversion Data Structure.} 
We now give our data structure, the \defn{all-function-inversion data structure}, which generalizes Lemma~\ref{lem:fiatnaorworst} to handle the above issues.  We first describe the data structure, and then prove its performance in Theorem~\ref{thm:fiatnaor_all}.

We describe how we can use sampling to find all points that a function $f: [N]\rightarrow D$ maps to a given value, and briefly outline the idea behind why our strategy is correct (the proof of Theorem~\ref{thm:fiatnaor_all} argues correctness more formally).  Along the way, we will handle the case where $f$ has large codomain.

Let's focus on inverting a single $f$ for a query $q$.
Assume momentarily that we are given $\kappa = |f^{-1}(q)|$ and that $1 < \kappa < o(N/\log N)$ (if $\kappa = 1$ then Lemma~\ref{lem:fiatnaorworst} suffices; if $\kappa = \Omega(N/\log N)$ then we can find $f^{-1}(q)$ by applying $f$ to all possible $N = \tilde{O}(|f^{-1}(q)|)$ elements in the domain).  Then we create $\Theta(\kappa \log N)$ sets, each obtained by sampling every element of $[N]$ with probability $1/\kappa$; denote these sets $\mathcal{N}^k$ for $k = 1\ldots \Theta(\kappa\log N)$.  

With high probability, since $N/\kappa = \Omega(\log N)$, each set $\mathcal{N}^k$ has $C_1N/\kappa$ elements for some constant $C_1$.\footnote{Assume $C_1$ is chosen so that $C_1N/\kappa$ is an integer.}  
Let us store all elements of $\mathcal{N}^k$ in an array of size $|\mathcal{N}^k|$---that way we can find the $i$th element of $\mathcal{N}^k$, denoted $\mathcal{N}^k[i]$, in $O(1)$ time (this requires an additional $\tilde{O}(N)$ space).
Then we can define a function $f_k: [C_1N/\kappa] \rightarrow [C_1N/\kappa]$ which simulates the behavior of $f$ on $\mathcal{N}^k$ as $f_k(i) = h_k(f(\mathcal{N}^k[i]))$ where $h_k: D\rightarrow [C_1N/\kappa]$ is from a universal family.
We build the data structure from Lemma~\ref{lem:fiatnaorworst} for each $f_k$.

Call an element $x\in f^{-1}(q)$ a \defn{singleton} for $\mathcal{N}^k$ if $x$ is the only element in $f^{-1}_k(q)$ (see also~\cite{KopelowitzPorat19}).  By standard Chernoff bounds (Corollary~\ref{cor:chernoff_log}), $x$ is a singleton for $\Theta(\log n)$ sets $\mathcal{N}^k$.  If $x$ is a singleton, then the data structure from Lemma~\ref{lem:fiatnaorworst}  returns $x$ with constant probability; again by Corollary~\ref{cor:chernoff_log}, $x$ is returned for some $\mathcal{N}^k$ with high probability.

We remove the assumption that $\kappa$ is known up front using repeated doubling.  We begin with $\kappa = 2$. We run the above algorithm 
for $f_k$ for $k = 1, \ldots, \kappa\log N$.  
Specifically, for each $k$, we query for a $j = f^{-1}_k(h_k(f(q)))$, look up $y = \mathcal{N}^k[j]$ in the array, and check that $f(y) = f(q)$.
If over all $k$ queries at least $\kappa$ distinct elements $x$ are found with $f(x) = f(q)$, we double $\kappa$ and repeat; otherwise we return all elements found so far as $f^{-1}(q)$.

We call the above data structure the \defn{all-function-inversion data structure}.

\begin{theorem}%
\label{thm:fiatnaor_all}
Consider the all-function inversion data structure built with parameter $\sigma$ for a set of $R \leq N$ functions $f_1, \ldots f_R$, where $f_i: [N] \rightarrow D$ and $f_i$ can be calculated in $T(f)$ time for all $i$.
This data structure requires $\tilde{O}(N +  NR/\sigma)$ space, 
can be built in $\tilde{O}(NR)$ time, 
and can
	find $f^{-1}_i(q)$ for any query $q\in [N]$ and any $f_i$ with high probability.  Each query requires 
	\[
		\tilde{O}\left(T(f) \left(1 + |f^{-1}_i(q)|\right)\sigma^3\right)
	\]
 expected time.
\end{theorem}

Before proving Theorem~\ref{thm:fiatnaor_all} we prove a useful intermediate lemma that treats the idea of a singleton in more generality.  
In the above discussion, we used $\kappa$ functions $f_k$ to invert some function $f$; from now on, 
we refer to these as the $k$ functions $f_{i,k}$ used to invert each function $f_i$.

\begin{lemma}%
\label{lem:singleton_general}
For any $q\in [N]$, any $\kappa \geq 2$, 
and any $k$ (with $f_i$ and $f_{i,k}$ as defined above),
for any 
set $X\subseteq f_i^{-1}(q)$ with $|f_i^{-1}(q) \setminus X| \leq 2\kappa$,
the probability that 
some element of $X$ is returned when querying the function inversion data structure for $f_{i,k}$ is $\Omega(1)$ if $|X| > \kappa/2$, and $\Omega(|X|/\kappa)$ if $|X| \leq \kappa/2$.
\end{lemma}

\iffull
\begin{proof}
	If all of the following four events occur, then some element of $X$ is returned when querying $f_{i,k}$:
	\begin{enumerate}
		\item $f_i^{-1}(q) \cap \mathcal{N}_k \subseteq X$, 
		\item $X \cap \mathcal{N}_k \neq \emptyset$, 
		\item no $y \in [N] \setminus f_i^{-1}(q)$ has  $f_{i,k}(y) = f_{i,k}(q)$, and
		\item the function inversion data structure correctly returns an element in $f^{-1}_{i,k}(h_k(q))$.
	\end{enumerate}
	These events are independent: the first depends only on what elements of $f_i^{-1}(q)\setminus X$ are sampled in $\mathcal{N}^k$; 
	the second depends on what elements of $X$ are sampled in $\mathcal{N}^k$ (with $X\subseteq f_i^{-1}(q)$);
	the third depends only on $h_k$ and what elements in $[N]\setminus f_i^{-1}(q)$ are sampled in $\mathcal{N}_k$; 
	the fourth depends on the internal randomness of the function inversion data structure.  We bound the probability of each event and multiply to obtain the lemma.  

	First event: let $Y = f^{-1}(q) \setminus X$; we have $|Y| \leq \kappa$.  
	The probability that $Y\cap \mathcal{N}^k = \emptyset$ is $(1 - 1/\kappa)^{|Y|} \geq (1 - 1/\kappa)^{2\kappa} \geq 1/16$.  

	Second event:\footnote{%
	This bound is well-known but we could not find a black-box reference that works for all $|X|$; the proof is included for completeness.} the probability that $X\cap \mathcal{N}_k \neq \emptyset$ is $1 - (1 - 1/\kappa)^{|X|}$.  This expression is increasing in $|X|$ and $\kappa$, so if $|X| \geq \kappa/2$ then since $\kappa \geq 2$ we have $1 - (1 - 1/\kappa)^{|X|} \geq 1/2$.  If $|X| < \kappa/2$,  
	{\allowdisplaybreaks
	\begin{align*}
		1 - (1 - 1/\kappa)^{|X|} &= 1 - \sum_{\ell = 0}^{|X|} \binom{|X|}{\ell}\left(\frac{-1}{\kappa}\right)^\ell\\
								 &\geq \frac{|X|}{\kappa} - \frac{(|X|-1)^2}{2\kappa^2} - \sum_{\ell = 4}^{|X|} \binom{|X|}{\ell}\left(\frac{-1}{\kappa}\right)^\ell \\
								 &\geq \frac{|X|}{\kappa} - \frac{|X|}{2\kappa} - \sum_{\ell' = 2}^{|X|/2} \binom{|X|}{2\ell'}\left(\frac{-1}{\kappa}\right)^{2\ell'} \\
								 &\geq \frac{|X|}{2\kappa} - \sum_{\ell' = 2}^{|X|/2} \left(\frac{e|X|}{2\ell'}\right)^{2\ell'}\left(\frac{-1}{\kappa}\right)^{2\ell'} \\
								 &\geq \frac{|X|}{2\kappa} - \sum_{\ell' = 2}^{|X|/2} \left(\frac{|X|}{2\kappa}\right)^{2\ell'} \\
								 &= \frac{|X|}{2\kappa} - \sum_{\ell'' = 0}^{|X|/2 - 2} \left(\frac{|X|}{2\kappa}\right)^{2\ell'' + 2} \\
								 &= \frac{|X|}{2\kappa} - \left(\frac{|X|}{2\kappa}\right)^2\sum_{\ell'' = 0}^{|X|/2 - 2} \left(\frac{|X|}{2\kappa}\right)^{2\ell''} \\
								 &\geq \frac{|X|}{2\kappa} - \frac{|X|}{8\kappa}\sum_{\ell'' = 0}^{\infty} \left(1/4\right)^{2\ell''} ~~\geq \frac{|X|}{4\kappa} 
\end{align*}
}
(Changing $\ell$ to $\ell'$ in the second step effectively removes all odd terms from the sum; in the fourth step we used $\binom{x}{y} \leq (ex/y)^{y}$.)

Third event: fix a $y\in [N]\setminus f_i^{-1}(q)$. 
The probability that $y\in \mathcal{N}^k$ is $1/\kappa$; let $y'$ satisfy $y = \mathcal{N}^k$.
Since $y\notin f_i^{-1}(q)$, if $f_{i,k}(y') = h_{i,k}(f_i(q))$, we must have $h_{i,k}(f_i(y)) = h_{i,k}(f_i(q))$ with $f_i(y) \neq f_i(q)$.  
By universality of $h_{i,k}$, $\Pr(h_{i,k}(f_i(y)) = h_{i,k}(f_i(q)) ~|~ y\in \mathcal{N}^k) \leq \kappa/(C_1 N)$.  
Multiplying, the probability that a fixed $y\in [N]\setminus f_i^{-1}(q)$ has $f_{i,k}(y) = f_{i,k}(q)$ is at most $1/(C_1 N)$.  Summing, in expectation we have $1/C_1$ such elements $y$; by Markov's inequality, the probability that there is at least $1$ such element is at most $1/C_1$.

Fourth event: occurs with constant probability by Lemma~\ref{lem:fiatnaorworst}.
\end{proof}
\else
	\begin{proof}
		See Appendix~\ref{sec:omitted_proofs}.
	\end{proof}
\fi

\begin{proof}[Proof of Theorem~\ref{thm:fiatnaor_all}]
	We begin with space and preprocessing time: we show that for each $\kappa$, the data structure requires $\tilde{O}(N/\sigma)$ space and $\tilde{O}(N)$ preprocessing time with high probability.   The function inversion data structure built on each $f_{i,k}$ requires $\tilde{O}(|\mathcal{N}^k|/\sigma)$ space and $\tilde{O}(|\mathcal{N}^k|)$ preprocessing time, plus $O(\log N)$ space to store $h_{i,k}$.  By Chernoff bounds (Corollary~\ref{cor:chernoff_log}), $\sum |\mathcal{N}^k| = O(N\log N)$ with high probability; summing over all $\log N$ values of $\kappa$ gives the bound.  

	Now, correctness. We split into two claims: first, that if $\kappa > |f_i^{-1}(q)|$, then (across all queries to $f_{i,k}^{-1}$) all elements in $f_i^{-1}(q)$ are found with high probability; second, that if $\kappa \leq |f_i^{-1}(q)|$, that $\kappa$ distinct elements from $f^{-1}(q)$ are found (and thus $\kappa$ is doubled and the search continues) with high probability.

	First, consider the case where $\kappa$ is the smallest power of $2$ satisfying $\kappa > |f^{-1}(q)|$.  
	Since there are fewer than $\kappa$ elements in $f^{-1}(q)$, the all-function-inversion data structure must return after this round; thus we are left to show that all elements of $f^{-1}(q)$ are returned in this round.  Fix an $x\in f^{-1}(q)$.  By Lemma~\ref{lem:singleton_general} with $X = \{x\}$ (and thus $|Y| = |f^{-1}(q)| \leq 2\kappa$), $x$ is returned by the data structure for $f_{i,k}$ with probability $\Omega(1/\kappa)$.  Since we choose a new $\mathcal{N}^k$ and $h_{i,k}$ for each $f_{i,k}$ these events are independent; thus by Corollary~\ref{cor:chernoff_log}, over all $\Theta(\kappa\log N)$ values of $k$, one returns $x$ with high probability.  Taking a union bound over all $< N$ values of $x$, all are returned with high probability.

	Now, consider $\kappa \leq f^{-1}_i(q)$.  A technical note on this case: the proof of in Lemma~\ref{lem:fiatnaorworst} from~\cite{FiatNaor00} does not as-is guarantee that a uniform random element is returned from $f^{-1}(q)$. A stronger version of this lemma that did provide such a guarantee would simplify this proof.  

	Partition the $\Theta(\kappa\log n)$ values of $k$ into \defn{epochs} of $C_2\kappa$ consecutive values.  There are $\Theta(\log n)$ epochs.  For the $j$th epoch, 
	let $F_j$ be the set of elements in $f^{-1}(q)$ that have been found by the data structure so far, and 
	let $r_j = \kappa - |F_j|$.  We say that the $j$th epoch is \defn{successful} if $r_{j+1} \leq r_j/2$.    After $\log \kappa + 2 = O(\log N)$ successful epochs, $\kappa$ elements have been found and we are done.   We now show that an epoch is successful with constant probability; a Chernoff bound (Corollary~\ref{cor:chernoff_log}) over the $\Theta(\log N)$ epochs gives the proof.

	Fix an epoch $j$ with $|F_j| < \kappa$; we want to show that it is successful with constant probability.  For any $k$ within epoch $j$, let $F$ be the elements found so far; assume that $F$ satisfies $\kappa - |F| > r_j/2$ (i.e.\ assume the epoch is not yet successful).  Note that since $r_j \geq 0$ this implies that $|F| < \kappa$.  Let $X \gets f^{-1}(q)\setminus F$; since $|f^{-1}(q)| \geq \kappa$ 
	we have that $|X| \geq \kappa - |F| > r_j/2$.  Applying Lemma~\ref{lem:singleton_general} to $X$ (with $F = f^{-1}(q)\setminus X$; notice $|F| < \kappa - r_j/2 \leq 2\kappa$), we have that a new element is found for a given $f_{i,k}$ with probability $\Omega(r_j/(2\kappa))$.  
	Consider the random variable representing if each $f_{i,k}$ increases $|F|$.  These are independent Bernoulli trials (they are independent since we choose a new $h_k$ and build a new function inversion data structure for each $f_{i,k}$); the expected number of elements found in the epoch is $\Omega(r_j)$, and the standard deviation is $\Theta(\sqrt{r_j})$; thus with constant probability the number of elements found is $\Omega(r_j)$ (i.e.\ is within one standard deviation of the expectation) and the epoch is successful.

	The time required is (with change of variable $\ell = \log_2 \kappa$ and defining $\lg x = \lfloor \log_2 \min\{x, 2\}\rfloor$):
	\[
		\sum_{\ell = 1}^{1 + \lg |f^{-1}(q)|} \tilde{O}(2^\ell \log N \sigma^3).
	\]
	Summing gives the desired bound.  
\end{proof}

\subsection{Applying Function Inversion to LSH}%
\label{sec:function_inversion_for_lsh}

We reiterate how LSH can help us solve ANN\@.  Consider a locality-sensitive hash family $\mathcal{L}$.  To solve ANN using $\mathcal{L}$, classically one selects\footnote{From now on we assume $p_1$ is a constant for $\mathcal{L}$ for simplicity; our results easily generalize.} $R = \Theta(n^{\rho})$ hashes $\ell_1\ldots, \ell_R$ from $\mathcal{L}$.  For each such $\ell_i$, one stores a reverse lookup table allowing us to find $x$ given $\ell_i(x)$, for all $x\in S$.  
	This requires $O(n)$ space per lookup table, giving $O(n^{1 + \rho})$ space overall. On a query, one iterates through $i  \in \{1,\ldots,R\}$, using the lookup table to find all $x\in S$ with $\ell_i(x) = \ell_i(q)$, giving $O(n^{\rho})$ expected query time.

Function inversion gives a space-efficient replacement for the lookup table---after all, the point of the lookup table is really to find $\ell_i^{-1}(\ell_i(q))$.  We now describe this strategy in more detail.

Consider 
$R = O(n^{\rho})$ hash functions $\ell_1, \ldots, \ell_R$ from a locality-sensitive hash family $\mathcal{L}$.
We define a new sequence of functions $\hat{\ell}_1, \ldots, \hat{\ell}_R$: for all $i\in [R]$ and $j\in [n]$, let $\hat{\ell}_i(j) = \ell_i(x_j)$.  (Thus, the domain of $\hat{\ell}_i$ is $[n]$ for all $i$.)

We apply Theorem~\ref{thm:fiatnaor_all} to invert all functions $\hat{\ell}_i$ with space savings $\sigma \gets n^s$.  The query algorithm follows immediately: for $i \in \{1,\ldots,R\}$, rather than looking up all values of $x$ with $\ell_i(x) = \ell_i(q)$ in the lookup table, we can instead query $\hat{\ell}^{-1}(\ell(q))$ to obtain the same candidate points, in $\tilde{O}(n^{3s})$ time per returned candidate point.

Thus, we store the following: the original set $S$, the hash functions $\ell_1,\ldots, \ell_R$, and an all-function-inversion data structure for each of $\hat{\ell}_1, \ldots, \hat{\ell}_R$.  We emphasize that we do not store the reverse lookup table for $\ell_1, \ldots, \ell_R$---we just store enough information to evaluate each function.

With this strategy we can obtain Theorem~\ref{thm:basic_lsh}.

\basiclsh*

\begin{proof}
	We build the all-function-inversion data structure with $\sigma = n^s$ for each $\ell_i$.  
	The space bound follows immediately.
	Correctness follows from correctness of the LSH---if a point $x$ with $d(x, q) \leq r$ has $\ell_i(x) = \ell_i(q)$ for some $i$  the all-function-inversion data structure returns $x$ on $\ell_i^{-1}(q)$ with high probability.  Thus, if $R$ was set so that a close point was found with probability $.9$, we find the close point with probability $.9(1 - 1/\poly(n))$.  Slightly increasing $R$ (doubling suffices) brings the probability above $.9$.

	The total query time is $\sum_i \tilde{O}(T n^{3s} (1 + |\ell_i^{-1}(q)|))$.  We have that $\E[\sum_i |\ell_k^{-1}(q)] = n^{\rho}$ by definition of the LSH and $\sum_i 1 = R = O(n^{\rho})$; substituting gives the lemma.
\end{proof}

\paragraph{Discussion.}
Setting $s = \rho$ gives a near-linear space data structure using any locality-sensitive hash family with query time $\tilde{O}(T n^{4\rho})$ where $T$ is the time to evaluate a function from the family.  Thus we obtain a black box near-linear space data structure with nontrivial query time for \emph{any} LSH family with $n^{o(1)}$ evaluation time, $O(n^{1-\rho})$ space,
and $\rho < 1/4$.  
We point out that most LSH families we are aware of easily meet these requirements: for example, a hash sampled from the classic Indyk-Motwani LSH family~\cite{IndykMotwani98} has $O(\log n)$ evaluation time and $O(\log n)$ space (as a hash from their family consists of $O(\log n)$ sampled dimensions).

For Euclidean ANN, Andoni and Indyk achieved an LSH with $\rho = 1/c^2$ and evaluation time $n^{o(1)}$~\cite{AndoniIndyk06}.  Applying Theorem~\ref{thm:basic_lsh} gives a data structure with $n^{1 + o(1)}$ space and query time $n^{4/c^2 + o(1)}$.  As seen in Table~\ref{tab:eucANNhistory} this is already competitive with many past space-efficient Euclidean ANN data structures.  

That said, the performance of this black box approach is worse than can be obtained by more recent results, e.g.~\cite{Christiani17,AndoniLaRa17}.  Improving on~\cite{Christiani17,AndoniLaRa17} requires more sophisticated techniques, which we give in Section~\ref{sec:faster_euclidean_linear_space_ann}.

\section{Near-Linear-Space Euclidean ANN}%
\label{sec:faster_euclidean_linear_space_ann}

In this section we use function inversion to improve the performance of 
the most performant near-linear-space ANN data structure for Euclidean distance.
Specifically, we improve performance of the space-efficient data structure for Euclidean ANN given by Andoni et al. in~\cite{AndoniLaRa17}; 
we call this data structure ALRW.

First, let us give some intuition for this improvement.
ALRW is able to achieve space $n^{1 + \rho_u + o(1)}$ and expected query time $n^{\rho_q + o(1)}$ for any $\rho_u, \rho_q$ satisfying
\begin{equation}
	\label{eq:alrw_l2}
	c^2\sqrt{\rho_q} + (c^2 - 1)\sqrt{\rho_u} \geq \sqrt{2c^2 - 1}.
\end{equation}
We can set $\rho_u = 0$ and obtain that ALRW achieves query time $n^{ALRW(c) + o(1)}$ with $ALRW(c) = (2c^2 - 1)/c^4$.  
Note that this is significantly better than we would get by applying Theorem~\ref{thm:basic_lsh} to these bounds directly: Equation~\ref{eq:alrw_l2} allows $\rho_q = \rho_u = 1/(2c^2 - 1)$, after which applying function inversion to obtain near-linear space would give query time $n^{4/(2c^2 - 1) + o(1)}$.

The key idea behind our improvement is to note that the query time of ALRW begins increasing quickly as $\rho_u$ approaches $0$.  For example, let's fix $c = 2$.  Then near-linear-space ALRW achieves query time $\approx n^{.44}$.  If we increase $\rho_u$ just slightly, the query time drops: setting $\rho_u = .01$ in Equation~\ref{eq:alrw_l2}, we obtain $\rho_q \approx .34$ (and thus query time $\approx n^{.34}$).  

This leaves room for function inversion to help:  consider instantiating ALRW with $\rho_q \approx .34$ and $\rho_u = .01$, with the plan to reduce the space to near-linear 
by applying Theorem~\ref{thm:basic_lsh} with $s = \rho_u$.  If this worked, we would incur total query time $n^{\rho_q + 4s} = n^{.38} \ll n^{.4375}$. 

In the rest of this section we give a more thorough exposition of these ideas.  In particular, there are several obstacles we must handle.  First, ALRW is not an LSH---it consists of a tree with a list of points at each leaf.    Therefore, Theorem~\ref{thm:basic_lsh} does not directly apply.
Moreover, the tree alone in ALRW requires $n^{1 + \rho_u + o(1)}$ space: thus, we do not even have enough space to store the internals of ALRW, even if we use function inversion to efficiently store the lists of points.
We must describe how to reduce the space usage of ALRW, as well as how function inversion can be applied to a tree rather than a hash to recover the lists of points at the leaves.  
Finally, we must set parameters: we want to find the value of $\rho_u$ to optimize performance.

In the rest of this section we combine function inversion with the ALRW data structure to obtain improved state-of-the-art query times for near-linear-space Euclidean ANN.

Throughout this section we will use $d(\cdot, \cdot)$ to refer to the Euclidean distance (we will extend our results to Manhattan distance using a standard reduction~\cite{AndoniLaRa17,LinialLoRa95}).  Our data structure begins with an instance of ALRW and modifies it to achieve our bounds. Throughout this section we use $\rho_u$ to denote the parameter used to construct ALRW.  (We will set the optimal $\rho_u$ for a given $c$ in Section~\ref{sec:analysis}.)

\subsection{The Optimal List-of-Points Data Structure}%
\label{sec:the_optimal_list_of_points_data_structure}

In this subsection we summarize ALRW.
We focus on the aspects of ALRW that are relevant to applying function inversion to their data structure; the reader should reference~\cite{AndoniLaRa16} for a full description and proof of correctness. Note that we use the full version of the paper~\cite{AndoniLaRa16}, as opposed to the conference version~\cite{AndoniLaRa17}, when referencing details of the algorithm.

ALRW is built using 
parameters 
$\rho_u$ and $\rho_q$ which provide the space/query time tradeoff: ALRW requires $n^{1 + \rho_u + o(1)}$ space, and queries can be performed in $n^{\rho_q + o(1)}$ time, so long as $\rho_u$ and $\rho_q$ satisfy Equation~\ref{eq:alrw_l2}.  

\subsubsection{High-level Description of ALRW}

ALRW consists of a single tree.  Each leaf of the tree has a pointer to a list of points; each internal node contains metadata to help with queries.  During preprocessing, the tree is constructed, along with the list for each leaf of the tree.  
We say that a leaf $\ell$ \defn{contains} a point $x$ if $x$ is in the list pointed to by $\ell$.

During a query $q$, a subtree of ALRW is traversed, beginning with the root.  
We say all nodes in this subtree are \defn{traversed by $q$}.
For each internal node traversed, in $n^{o(1)}$ time it is possible to find the children of the node that are recursively traversed for $q$.  When reaching a leaf, for each point in the list pointed to by the leaf, $d(x,q)$ is calculated. If $d(x,q) \leq cr$, then $x$ is returned.

A similar process defines what leaf contains each $x\in S$: a subtree of ALRW is traversed using the metadata stored at each internal node (in $n^{o(1)}$ time per node); each time a leaf is traversed it contains $x$.\footnote{In reality, the ALRW tree is built and the lists are created simultaneously using a single recursive process. However, it is useful for our exposition to imagine each point as if it were stored using a separate process after the tree was already constructed.}  We say that any node in this subtree is \defn{traversed by $x$}.

\subsubsection{ALRW Details and Parameters}%
\label{sec:alrw_details_and_params}

The internal ALRW nodes are of two types: \defn{ball nodes} and \defn{sphere nodes};
these correspond to recursive calls during construction to \textsc{ProcessBall} (for dense clusters) and \textsc{ProcessSphere} (for recursing using LSH) in~\cite{AndoniLaRa16}.\footnote{These two types of nodes are the reason why ALRW is data-dependent; they are integral to the bounds achieved by ALRW.}
Any node in the tree has at most $b_b$ ball node children, with $b_b = (\log\log\log n)^{O(1)}$, and at most $b_s$ sphere node children, with $b_s = n^{o(1)}$.

Consider a sphere node $v$ traversed when preprocessing a point $x\in S$. The number of children of $v$ that are also sphere nodes, and the probability that $x$ traverses each child, vary as we traverse the tree---however, they are set so that the expected number of sphere node children that $x$ traverses recursively remains the same.  
Let $b_u$ be this quantity: the expected number of sphere node children traversed when determining what lists to use to store $x$.
	By the proof of~\cite[Claim 4.12]{AndoniLaRa16}, 
	for any sphere node,
	$b_u = n^{(\rho_u + o(1))/K}$.  

The number of sphere nodes on a root-to-leaf path in ALRW is at most $K$ 
with $K = \Theta(\sqrt{\log n})$.  
Furthermore, all root-to-leaf paths have at most $K_b = \tilde{O}(\log\log n)$ ball nodes~\cite[Lemma 4.2]{AndoniLaRa16}.%
\footnote{Technically we are using $K$ slightly differently than in~\cite{AndoniLaRa16}---they use $K$ to bound the number of edges along a root-to-leaf path where both parent and child are sphere nodes; since $K_b = o(K)$ this does not affect our analysis.}

	In fact, if $\rho_u$ is a positive constant that does not depend on $n$, the number of children traversed by any $x$ is very close to $b_u$ with high probability. (This result does not appear explicitly in~\cite{AndoniLaRa16}, but can be obtained by applying Chernoff bounds to their analysis.)
\begin{lemma}%
\label{lem:branching_concentration}
Consider an ALRW data structure with parameter $\rho_u > 0 $ with $\rho_u= \Omega(1)$.
For any sphere node $v$ traversed when considering a point $x$, the number of sphere node children of $v$ also traversed when considering $x$ is at most $b_u(1 + 1/K)$ with high probability.  
\end{lemma}
\begin{proof}
	Whether $x$ is stored in each child depends on how close it is to a randomly-chosen vector stored for each child (see~\cite[Section 4.2]{AndoniLaRa16}).  These vectors are chosen independently; therefore, the number of children traversed is the sum of independent Bernoulli trials with expectation $b_u$.

	By Chernoff bounds (the first bound in Lemma~\ref{lem:chernoff}), the probability that more than $b_u(1 + 1/K)$ children are traversed is at most $e^{-b_u/(K^2(2 + 1/K))}$.  
	Since $b_u = n^{\rho_u/K}$ and $K = O(\sqrt{\log n}$), we have $b_u/(K^2(2+K)) = n^{\rho_u/\sqrt{\log n} - o(1)} = \Omega(\log n)$ (the last step is a very loose lower bound since $\rho_u = \Omega(1)$).
\end{proof}

	The candidate points for a given $q$ consist of the points contained by all leaves traversed by $q$.  Thus, the following lemma immediately gives correctness.
\begin{lemma}[\cite{AndoniLaRa16}, Lemma 4.9]%
\label{lem:alrw_correct}
For any query $q$, let $x\in S$ be a point with $d(q,x) \leq r$.  
Then with probability $> .9$, some leaf $\ell$ of ALRW traversed by $q$ is traversed by $x$.
\end{lemma}

\subsection{Making the Tree Space-Efficient}%
\label{sec:making_the_tree_space_efficient}

ALRW requires $n^{1 + \rho_u(c) + o(1)}$ space, even if we ignore the lists of points at each leaf (i.e., there are $\Omega(n^{1 + \rho_u(c)})$ nodes in the tree).  
Furthermore, this high space usage appears to be integral to the data structure: 
since the tree is recursively data-dependent, each tree node stores extra information about the data being stored.  In this section we describe how to reduce the space of ALRW while still retaining enough data to answer queries efficiently.

\newlength{\outertreewidth}
\setlength{\outertreewidth}{2cm}
\newlength{\outertreeheight}
\setlength{\outertreeheight}{4cm}

\begin{figure}[t]
\begin{center}
	\subfloat[We begin with ALRW (on the left), and remove all nodes with with more than $\bar{K}$ ancestor sphere nodes to obtain the middle step tree (on the right).]{%
\begin{forest}
	dottededge/.style={edge=densely dashed},
	thickedge/.style={edge=very thick},
  for tree={%
    inner sep=0pt,
    minimum size=3pt,
    circle,
    fill,
	edge = thick
  },
	[,name=root
			[[]
				[]
				[ 
					[
						[
						[[][,phantom]]
						[[][]]
						]
						[]
						[]
					] []
				]
			]
			[
				[]
				[ 
					[
						[]
						[
						[[][]]
						[]
						]
					] []
				]
			]
	]
	\draw[dashed] (2, -4) -- (-2,-4);
	\node at (2, -4) (scissor) {\ScissorHollowLeft};
	\draw[dotted,<->] (2,0) -- (scissor) node[midway,fill=white]{$\bar{K}$};
\end{forest}
\hspace{.6in}
\begin{forest}
	dottededge/.style={edge=densely dashed},
	thickedge/.style={edge=very thick},
	bigleaf/.style={inner sep=2pt,fill},
	smallleaf/.style={inner sep=1pt,thick,draw},
  for tree={%
    inner sep=0pt,
    minimum size=3pt,
    circle,
	fill,
	edge = thick
  },
	[,name=root
			[[]
				[]
				[ 
					[[][][]]
					[]
				]
			]
			[
				[]
				[ 
				[ [] [[,phantom[,draw=none,edge=none,fill=none]]  ]
					] []
				]
			]
	]
\end{forest}%
\label{fig:middle_tree}%
}
\hspace{.8in}
	\subfloat[%
Recursively removing small leaves (outlined) of the middle step tree (on the left) gives the truncated tree (on the right).]{%
\begin{forest}
	dottededge/.style={edge=densely dashed},
	thickedge/.style={edge=very thick},
	bigleaf/.style={inner sep=2pt,fill},
	smallleaf/.style={inner sep=2pt,draw,fill=none},
  for tree={%
    inner sep=0pt,
    minimum size=3pt,
    circle,
	fill,
	edge = thick
  },
	[,name=root
			[[,bigleaf]
				[,smallleaf]
				[ 
					[,name=trimmed
					 [,smallleaf] [,smallleaf] [,smallleaf]
					] 
					[,bigleaf]
				]
			]
			[
				[,bigleaf]
				[ 
				[ [,bigleaf] [,bigleaf[,phantom[,draw=none,edge=none,fill=none]]]
					] [,bigleaf]
				]
			]
	]
	\node[below left=-5pt and -3pt of trimmed,rotate=-15] {\ScissorHollowRight};
\end{forest}
\hspace{.7in}
\begin{forest}
	dottededge/.style={edge=densely dashed},
	thickedge/.style={edge=very thick},
	bigleaf/.style={inner sep=2pt,fill},
	smallleaf/.style={inner sep=1pt,thick,draw},
  for tree={%
    inner sep=0pt,
    minimum size=3pt,
    circle,
	fill,
	edge = thick
  },
	[,name=root
			[[]
				[]
				[ 
					[,name=trimmed
					] 
					[]
				]
			]
			[
				[]
				[ 
				[ [] [[,phantom[,draw=none,edge=none,fill=none]]]
					] []
				]
			]
	]
\end{forest}%
\label{fig:truncated_tree}%
\hspace{.1in}
}
\end{center}%
\caption{Creating the truncated tree in two steps.}
\vspace{-.1in}
\end{figure}
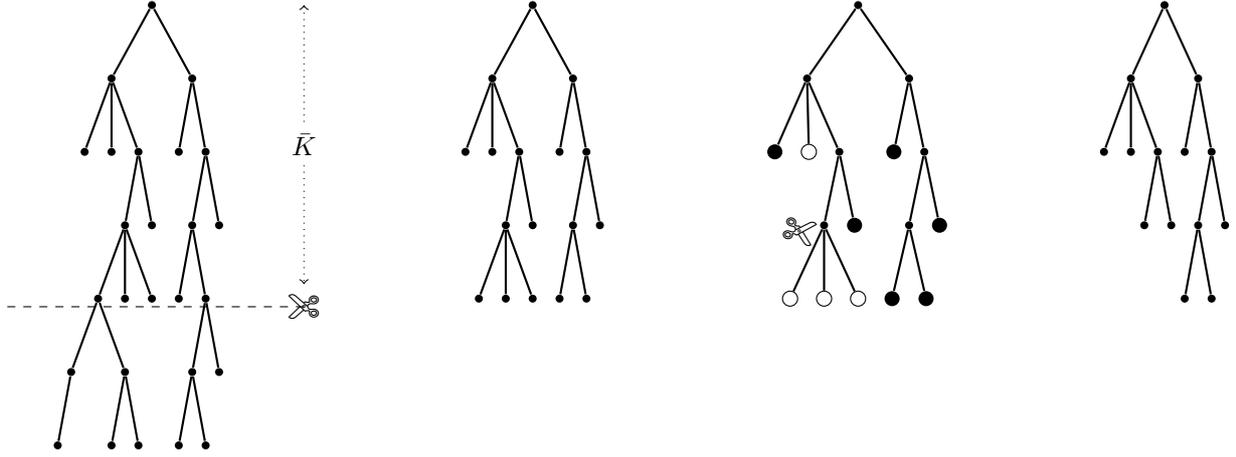

We create a \defn{truncated tree} which is a subtree of ALRW\@.
The main idea is as follows.  ALRW consists of $n^{1 + \rho_u + o(1)}$ leaves, each of which contains $O(1)$ points on average.  The goal of the truncated tree is to have $n^{1 + o(1)}$ leaves, each with $O(n^{\rho_u})$ points on average.  
The truncated tree is created carefully to satisfy some slightly stronger properties: 
we also want to bound the number of sphere nodes along the path from the root to any leaf in the tree (for e.g.\ the proof of Lemma~\ref{lem:number_functions}); furthermore, we want to be sure that any leaf traversed by a query contains $O(n^{\rho_u})$ expected points (for e.g\ the proof of Lemma~\ref{lem:leaf_sizes}).

We create the truncated tree in two steps.
First, let $\overline{K} = K/(1 + \rho_u)$ and consider the largest subtree of ALRW such that there are at most $\overline{K}$ sphere nodes on any root-to-leaf path of the truncated tree.  Call this subtree the \defn{middle step tree}; see Figure~\ref{fig:middle_tree}.  

We can define lists for each leaf of the middle step tree by running the recursive process from ALRW for each $x\in S$.  (Since this process is recursive, the points in the list of in a middle step tree leaf are exactly the points in its descendant leaves in ALRW.)

Now, we trim the middle step tree to obtain the truncated tree.
We proceed bottom-up through the middle step tree, starting with the second-to-last level.  
Call a leaf of the middle step tree \defn{large} if it contains at least $n^{\rho_u/(1 + \rho_u)}$ points and \defn{small} otherwise.
If all children of a node $v$ are small leaves, we remove all children of $v$ (so $v$ is now a leaf), repeating until no nodes have only small leaves as children.
We call this structure the \defn{truncated tree}.

The following lemma shows that the truncated tree achieves near-linear space.
\begin{lemma}%
\label{lem:truncated_tree_space}
The truncated tree contains $n^{1 + o(1)}$ nodes in expectation, each of which requires $n^{o(1)}$ space.
\end{lemma}
\begin{proof}
	Each ALRW tree node requires $n^{o(1)}$ space; this applies to the truncated tree nodes as well.

	We begin by bounding the size of all lists of the middle step tree (while we are ultimately ignoring the size of the lists in this result, it will help us analyze the number of nodes).  By~\cite[Claim 4.12]{AndoniLaRa16}, each $x\in S$ appears in $n^{\rho_u \overline{K}/K}$ lists in the middle step tree in expectation.  Thus, the expected size of all lists in the middle step tree is $n^{1 + \rho_u\overline{K}/K} = n^{1 + \rho_u/(1 + \rho_u)}$.

	The total size of all leaves in the truncated tree is at most the total size of all leaves in the middle step tree (because if $x$ traverses $v$ it must traverse all ancestors of $v$); thus the total size of all leaves in the truncated tree is at most $n^{1 + \rho_u/(1 + \rho_u)}$ in expectation.  Since each leaf of the truncated tree contains at least $n^{\rho_u/(1 + \rho_u)}$ elements, the number of large leaves is at most $O(n)$ in expectation.  Every internal node is the ancestor of some large leaf, so there are at most $n^{1 + o(1)}(\overline{K} + K_b))$ internal nodes. Any small leaf is the sibling of a large leaf or an internal node, so there are at most $O(n^{1 + o(1)}(b_s + b_b)(1 + \overline{K} + K_b))$ small leaves.  
	Bounding each term and summing, there are $n^{1 + o(1)}$ expected total nodes.  
\end{proof}

Our data structure stores a distinct label (e.g.\ from $1$ to $n^{1 + o(1)}$) for each node of the truncated tree.  We occasionally refer to a leaf using its label (if $\ell$ is a label we may say ``leaf $\ell$'').
We retain the definition that a leaf $\ell$ \defn{contains} $x \in S$ in the truncated tree if $x$ traverses $\ell$, even though we do not store the actual lists.

Note that Lemma~\ref{lem:alrw_correct} immediately applies to the truncated tree: if $d(x,q) \leq cr$, then with probability $\geq .9$, $x$ is contained in a leaf of the truncated tree traversed by $q$.

\subsection{Applying Function Inversion}%
\label{sec:applying_function_inversion}

\subsubsection{Truncated tree functions.}%
\label{sec:truncated_tree_functions}

The job of the truncated tree functions is to recover the points contained by each leaf in the truncated tree. Specifically if $x_j$ is contained in some leaf with label $\ell$, we want there to be a truncated tree function $\tau$ such that $\tau(j) = \ell$; thus, $\tau^{-1}(\ell) = j$.

The challenge is that each $x\in S$ maps to many points in the truncated tree, whereas each function must output a single leaf.  Fortunately, we have shown that the way $x$ traverses the tree is highly regular: if a node $v$ is traversed by $x$, with high probability, the number of children of $v$ traversed by $x$ is at most $b_u(1 + 1/K)$. This allows us to efficiently iterate over the leaves traversed by $x$.

A \defn{tree route} $I$ is
a vector of length $\bar{K} + K_b$ where each entry in $I$ is a number between $1$ and $b_u(1 + 1/K) + b_b$.  
Let $R = (b_u(1 + 1/K) + b_b)^{\bar{K} + K_b}$ be the number of tree routes.

\begin{lemma}%
\label{lem:number_functions}
The number of possible tree routes is $R = n^{\rho_u/(1 + \rho_u) + o(1)}$.
\end{lemma}
\iffull
\begin{proof}
The number of possible $I$ is $R = (b_u(1 + 1/K) + b_b)^{\bar{K} + K_b}$.  
Since $b_u, b_b > 2$ 
we can loosely upper bound $R \leq (b_u(1 + 1/K))^{\bar{K} + K_b} \cdot b_b^{\bar{K} + K_b}$.

Recall that $b_b = (\log\log\log n)^{O(1)}$, $\bar{K} = K/(1 + \rho_u) = O(\sqrt{\log n}/(1 + \rho_u))$, and $K_b = \tilde{O}(\log\log n)$.  We immediately have that $b_b^{\bar{K} + K_b} = n^{o(1)}$, and $(1 + 1/K)^{\bar{K} + K_b} = n^{o(1)}$.

We have (see Section~\ref{sec:alrw_details_and_params}) $b_u = n^{\rho_u + o(1)}/K$.  Therefore, $b_u^{\bar{K}} = n^{\rho_u/(1 + \rho_u) + o(1)}$.  Similarly, we have $b_u^{K_b} = n^{\rho_u/\sqrt{\log N} + o(1)} = n^{o(1)}$.
\end{proof}
\else
	\begin{proof}
		See Appendix~\ref{sec:omitted_proofs}.
	\end{proof}
\fi

For all possible tree routes $I$, we define a function $\tau_I(j)$ using the following process.  
We begin with vertex $v$ equal to the root of the truncated tree.  In step $i$ (for $i$ from $1$ to $\bar{K}$), we set $v$ equal to the $I[i]$th child of $v$ traversed by $x_j$.  (There are at most $b_u(1 + 1/K) + b_b$ such children by Lemma~\ref{lem:branching_concentration} and definition of $b_b$.)  If at any point there are less than $I[i]$ children of $v$ traversed by $x_j$, then $\tau_I(j) = -1$.  If $v$ is a leaf with label $\ell$, then $\tau_I(j) = \ell$.

We immediately obtain the following.  First, if $x$ is contained in $\ell$, then there is a tree route $I$ with $\tau_I(j) = \ell$.  
Second, for any $j\in [N]$ and any tree route $I$, we can calculate $\tau_I(j)$ in $n^{o(1)}$ time.  (This follows from the definition of ALRW: each node has $n^{o(1)}$ total children which we consider one by one, and for each child we can determine if $x$ traverses that child in $n^{o(1)}$ time; multiplying by $K$ retains $n^{o(1)}$ time.)

\subsubsection{The Data Structure}
First, preprocessing.  We begin by creating the truncated tree.  
Then, for all $R$ possible tree routes $I$, we build the all-function-inversion data structure (Theorem~\ref{thm:fiatnaor_all}) on $\tau_I$ with $\sigma = R$.  
This data structure requires $n^{1 + o(1)}$ space (noting that the domain of all $\tau_I$ is $[n]$).

To query, 
we use the function inversion data structure to recover the contents of all leaves traversed by $q$ in the truncated tree as follows.  We calculate the labels of leaves of the truncated tree traversed by $q$; call this set of labels $L$.  For each $\ell\in L$ and each tree route $I$, we query $\tau_I^{-1}(\ell)$.  We compare $q$ to each point returned; if any has distance at most $cr$ from $q$, it is returned.  If no such point is found, we return that there is no close point.

\subsubsection{Analysis}
\label{sec:analysis}

Correctness follows immediately from ALRW: if $d(q,x) \leq r$, then with probability $\geq .9$, $x$ is in a leaf traversed by $q$ in the truncated tree.  The all-function-inversion data structure will return this leaf with high probability. 

Now we give performance.  The following is the key performance lemma, bounding the number of leaves traversed by the query in expectation, as well as the total number of points at distance $>cr$ they contain.

\begin{lemma}%
\label{lem:leaf_sizes}
The number of leaves of the truncated tree traversed by $q$ is at most $n^{\rho_q/(1 + \rho_u) + o(1)}$ in expectation.
The expected number of points $x$ satisfying: (1) $x$ is contained in a leaf traversed by $q$, and (2) $d(x,q) > cr$, is at most $n^{\rho_u/(1 + \rho_u) + o(1)}$. 
\end{lemma}
\begin{proof}
By the proof of~\cite[Lemma 4.13]{AndoniLaRa16},\footnote{Specifically, by setting the base case using $l = \overline{K}$ rather than $K$ in Equation (16) of~\cite{AndoniLaRa16} and solving.} $q$ traverses $n^{\rho_q \overline{K}/K + o(1)} = n^{\rho_q/(1 + \rho_u) + o(1)}$ expected leaves in the middle step tree.  

By the proof of~\cite[Lemma 4.13]{AndoniLaRa16} (substituting $\overline{K}$ for $K$ in the exponent of the last equation), 
	we have that the expected number of total points with distance $\geq cr$ stored in all lists of leaves of the middle step tree traversed by $q$ is bounded by
	\iffull
	\[
	n^{1 + (\rho_q - 1)/(1 + \rho_u) + o(1)} = n^{(\rho_u + \rho_q)/(1 + \rho_u) + o(1)}.
	\]
\else
	$n^{1 + (\rho_q - 1)/(1 + \rho_u) + o(1)} = n^{(\rho_u + \rho_q)/(1 + \rho_u) + o(1)}$.
\fi
	Some leaves of the truncated tree are leaves of the middle step tree; their total size is bounded by the above.  

If a leaf in the truncated tree is an internal node of the middle step tree, its children all have size at most $n^{\rho_u/(1 + \rho_u)}$.  Since any node has $n^{o(1)}$ children, any leaf of the truncated tree that is not a leaf of the middle step tree contains $n^{\rho_u/(1 + \rho_u) + o(1)}$ points.  Multiplying by the number of leaves in the truncated tree traversed by $q$ gives a total of $n^{(\rho_q + \rho_u)/(1 + \rho_u) + o(1)}$ points.
\end{proof}

\begin{lemma}%
\label{lem:performance}
The above data structure requires $n^{1 + o(1)}$ space, preprocessing time $nR = n^{1 + \rho_u/(1 + \rho_u) + o(1)}$, and expected query time 
	$n^{(\rho_q + 4\rho_u)/(1 + \rho_u)}$.
\end{lemma}
\begin{proof}
	The truncated tree requires a total of $n^{1 + o(1)}$ space.  The time to build the middle step tree and truncated tree is $n^{1 + \rho_u/(1 + \rho_u) + o(1)}$.  (We cannot build all of ALRW and truncate it in this time---instead, we must build the middle step tree recursively as in~\cite{AndoniLaRa16} with maximum depth $\bar{K}$.)  Finally, preprocessing the all-function-inversion data structure requires $\tilde{O}(nR)$ time by Theorem~\ref{thm:fiatnaor_all}.

	Recalling that a node can be traversed in $n^{o(1)}$ time, we can find the set of leaves $L$ traversed by $q$ in $n^{\rho_q/(1 + \rho_u) + o(1)}$ time by Lemma~\ref{lem:leaf_sizes}.

Now, we bound the time to complete the function inversion queries.
By Theorem~\ref{thm:fiatnaor_all}, for any $I$ and any leaf $\ell$ of the conceptual tree, we can find $\tau_i^{-1}(\ell)$ in $\tilde{O}(R^3(|\tau_I^{-1}(\ell)| + 1))$ expected time.

Summing over all $\ell$ and $I$ and collecting terms, the query time is
\iffull
\[
	\tilde{O}\left(R^4|L| + R^3\sum_{\ell\in L, I} |\tau_I^{-1}(\ell)|\right).
\]
\else
	$\tilde{O}\left(R^4|L| + R^3\sum_{\ell\in L, I} |\tau_I^{-1}(\ell)|\right)$.
\fi
Note that when we find a point at distance $\leq cr$ we return it; thus we need only include points at distance $>cr$ in each $\tau_I^{-1}(\ell)$ term.  

By Lemma~\ref{lem:number_functions}, $R^3|L| = n^{(\rho_q + 3\rho_u)/(1 + \rho_u) + o(1)}$.  By Lemma~\ref{lem:leaf_sizes}, 
\[
	\E\left[\sum_{\ell\in L, I} |\tau_I^{-1}(\ell)|\right] = n^{(\rho_u + \rho_q)/(1 + \rho_u) + o(1)}
\]

Summing we obtain the lemma.
\end{proof}

Now, we set $\rho_u$ to optimize the above results.  Note that the following bounds are somewhat loose.  For query time, we ignore the $1 + \rho_u$ term in the denominator of the exponent.  For preprocessing time, we give the worst case for any $c$; if $c$ is small or large the bound is fairly pessimistic.  That said, both of these losses have, ultimately, a modest effect on the running time since we will choose small values for $\rho_u$.

We use $\alpha(c)$ to bound our running time (rather than the more traditional $\rho$) to prevent confusion with $\rho_u$ and $\rho_q$ used by ALRW, and $\rho$ used in Theorem~\ref{thm:basic_lsh}.

\runningtime*

\begin{proof}
	By Lemma~\ref{lem:performance}, the query time of our data structure is $n^{(\rho_q + 4\rho_u)/(1 + \rho_u) + o(1)}$.  We can upper bound this with $n^{\alpha(c) + o(1)}$, with $\alpha(c) = \min_{\rho_q, \rho_u} \rho_q + 4\rho_u$.  

Substituting $\rho_q$ using Equation~\ref{eq:alrw_l2}, we obtain
\[
	\alpha(c) = \min_{\rho_u} \frac{\left(\sqrt{2c^2 - 1} - (c^2-1)\sqrt{\rho_u}\right)^2}{c^4}  + 4\rho_u.
\]
Rearranging terms,
\[
	\alpha(c) = \min_{\rho_u} \rho_u \cdot(4 + (c^2 - 1)^2/c^4) - \sqrt{\rho_u}\cdot 2(c^2 - 1)\sqrt{2c^2 - 1}/c^4 + (2c^2 - 1)/c^4.
\]
Any quadratic $Ax^2 + Bx + C$ with $A,B > 0$ has minimum value $C - B^2/(4A)$ at $x = -B/(2A)$.  Substituting and simplifying, we obtain the claimed bound on $\alpha(c)$.
\begin{align*}
\alpha(c) &= \frac{2c^2 -1}{c^4} - \frac{1}{4}\left(\frac{4(c^2-1)^2(2c^2 -1)}{c^8}\right)\left(\frac{1}{4 + (c^2 - 1)^2/c^4}\right)\\
			  &= \frac{2c^2 - 1}{c^4}\left(1 - \frac{(c^2 - 1)^2}{4c^4 + (c^2 - 1)^2}\right).
\end{align*}

The preprocessing time is at most $n^{1 + \rho_u/(1 + \rho_u) + o(1)}$.  In the above we set $\rho_u$ to satisfy 
	\[
		\sqrt{\rho_u} = \sqrt{2c^2 - 1}\cdot \frac{(c^2 -1)}{4c^4 + (c^2 - 1)^2}.
	\]
	It is possible to verify numerically that 
	for any $c$, the $\rho_u$ satisfying the above has $\rho_u \leq .013$ (the closed form of the maximum appears to be too complicated to give useful intuition).
\end{proof}

\paragraph{Manhattan ANN}

As mentioned in Section~\ref{sec:introduction}, we can obtain bounds for Manhattan ANN by a classic embedding; see~\cite{LinialLoRa95}.  This immediately gives a Manhattan ANN data structure with query time $n^{\alpha_M(c) + o(1)}$ with 
\[
	\alpha_M(c) = \frac{2c - 1}{c^2}\left(1 - \frac{(c - 1)^2}{4c^2 + (c - 1)^2}\right).
\]
The preprocessing time remains bounded above by $n^{1.013 + o(1)}$.  The proof is essentially identical to that of Theorem~\ref{thm:running_time}.

\bibliographystyle{plain}
\bibliography{spaceTimeLSH}

\end{document}